\documentclass[sigplan,screen,nonacm,10pt]{acmart}

\AtBeginDocument{%
  }

\setcopyright{acmlicensed}
\copyrightyear{2018}
\acmYear{2024}
\acmDOI{XXXXXXX.XXXXXXX}

\acmConference[SPLASH 2024]{Proceedings of the 2023 ACM SIGPLAN International Symposium on New Ideas, New Paradigms, and Reflections on Programming and Software}{October 20--25, 2024}{Pasadena, California, USA}
\acmISBN{978-1-4503-XXXX-X/18/06}

\graphicspath{ {./images/} }
\usepackage{subcaption}
\usepackage{algorithm}
\usepackage{algpseudocode}
\usepackage{todonotes}
\usepackage{ulem}

\setcopyright{none}

\begin{document}

\title{Obfuscation as Instruction Decorrelation}

\author{Ali Ajorian}
\affiliation{%
 \institution{University of Basel}
 \country{Switzerland}}
\email{ali.ajorian@unibas.ch}

\author{Erick Lavoie}
\affiliation{%
 \institution{University of Basel}
 \country{Switzerland}}
\email{erick.lavoie@unibas.ch}

\author{Christian Tschudin}
\affiliation{%
 \institution{University of Basel}
 \country{Switzerland}}
\email{christian.tschudin@unibas.ch}


\begin{abstract}
Obfuscation of computer programs has historically been approached either as a practical but \textit{ad hoc} craft to make reverse engineering subjectively difficult, or as a sound theoretical investigation unfortunately detached from the numerous existing constraints of engineering practical systems.

In this paper, we propose \textit{instruction decorrelation} as a new approach that makes the instructions of a set of real-world programs appear independent from one another. We contribute: a formal definition of \textit{instruction independence} with multiple instantiations for various aspects of programs; a combination of program transformations that meet the corresponding instances of instruction independence against an honest-but-curious adversary, specifically random interleaving and memory access obfuscation; and an implementation of an interpreter that uses a trusted execution environment (TEE) only to perform memory address translation and memory shuffling, leaving instructions execution outside the TEE.

These first steps highlight the practicality of our approach. Combined with additional techniques to protect the content of memory and to hopefully lower the requirements on TEEs, this work could potentially lead to more secure obfuscation techniques that could execute on commonly available hardware.
\end{abstract}



\received{25 April 2024}

\maketitle

\section {Introduction}

Obfuscation consists in transforming a program to hinder the extraction of valuable information, such as algorithms, data structures, and secret keys, either from its static representation, e.g. source code, or its dynamic behaviour, i.e. the sequence of states traversed during its evaluation, while maintaining the program's original behavior. Obfuscation has numerous applications including software protection \cite{collberg1997taxonomy}, digital watermarking \cite{collberg2002watermarking} and cryptographic application such as public key encryption, secure signatures, zero knowledge proofs \cite{gonzalez2023software,sahai2014use}, etc. Research on obfuscation historically followed two main directions, one empirical and one theoretical. 

On the one hand, since the 90s~\cite{xu2020layered}, researchers have devised ad-hoc heuristics that aim to make the internal structure of software unintelligible to reverse engineers. Although these techniques are widely used and at least partially effective in practice, the exact security properties they provide is not clear, especially since the adversary against which they protect is not well (objectively) defined.  

On the other hand, the existence of a general obfuscator, which would prevent any leakage of internal structure by only revealing input-output behaviour (i.e. a \textit{virtual black-box}~\cite{barak2001possibility}), has been disproved because there exists certain classes of functions whose input-output behaviour does reveal their internal structure, regardless of the scheme used for obfuscation~\cite{barak2001possibility}. For instance, it is easy to guess the missing parameters of polynomial functions in polynomial time by testing a limited number of inputs, e.g. the $m$ and $b$ parameters of the function $y=m*x+b$ can be derived with only two pairs of input-outputs.\footnote{Another example is Quine programs: they take no input and produce a copy of their own source code as their only output and therefore trivially defeat obfuscation.}  Following this result, a number of other alternative approaches have been investigated including \textit{indistinguishability obfuscation}~\cite{barak2001possibility}, \textit{best possible obfuscation}~\cite{goldwasser2007best} and \textit{Grey box obfuscation}~\cite{kuzurin2007concept}, that are all designed as mathematical constructs on Turing machines or gate-level circuits. Directly simulating these approaches on general purpose computers incurs a significant slowdown and most of the details required to turn them into practical techniques are left as future "engineering" work. 

In this paper, we propose a new obfuscation approach, that aims at a middle-ground by leveraging the best of both previous directions. On the one hand, we take an empirical approach to obfuscation: we start with a set of source codes that correctly and efficiently execute on common hardware; we then combine and transform multiple of these programs into a single new program that would still execute on common hardware but we run it through a specially-designed interpreter. Sticking to the programs' original instructions ensures the obfuscation will be practical and not hopelessly slow. On the other hand, we aim to be robust against an adversary that commands a large but finite number of classical computers, i.e. a Probabilistic Polynomial-Time (PPT) adversary, and who could use those computers to perform static and dynamic analyses on the obfuscated programs. In order to achieve this, we design (a) a source transformation and (b) an interpreter such that \textit{each executed instruction appears independent of any other instruction to the adversary}, an approach we call \textit{instruction decorrelation}. Because correlating instructions is the first necessary step in reconstructing source programs, and meaningless random instructions can be interleaved with those of the source programs, this security property turns reverse-engineering into an NP-hard problem whose difficulty is proportional to the number and size of source programs, as well as the amount of random instructions that were injected during obfuscation.

A structurally similar approach of hiding correlation from an attacker is used in Rivest's \textit{Chaffing and Winnowing} technique from 1998~\cite{rivest1998chaffing}: Instead of producing a ciphertext, a message's bits are sent in clear but burried in a stream of other bits, making them indistinguishable from the chaff. In analogy to Rivest's point that confidentiality can be obtained without encryption, we strive for strong obfuscation despite leaving the instructions intact.

Towards the goal of making our approach practical, we make the following contributions: 1) we articulate \textit{instruction independence} as a new unintelligibility property around which to design practical obfuscators, with a formal cryptographic definition that provides security against a PPT adversary; 2) we describe two program transformations that respectively interleave program execution and remove memory access dependencies, making obfuscated source instructions appear independent; 3) we provide an interpreter design that provides instruction independence against an honest-but-curious adversary, does not require the code of obfuscated programs to be executed within a trusted execution environment (TEE), and only leverages a TEE to perform memory address translation and periodic memory shuffling.

By itself, our approach is not sufficient to obtain a complete and perfect obfuscation because it does not protect the content of memory during execution, which could be regularly scanned to extract, e.g., secrets such as the private keys used for asymmetric encryption. Nonetheless, it is a first step in the direction of provably-secure yet practical obfuscation. We believe our current implementation could be combined with additional techniques, such as encryption, to obtain a complete secure execution environment in the future.

The rest of this paper is structured as follows: Section~\ref{sec:preliminaries} introduces our notation conventions, adversary model, as well as background and related work; Section~\ref{sec:obfISint} provides a formal definition of instruction decorrelation as well as several properties that it implies; Section~\ref{sec:construction} presents an implementation designed according to the properties previously derived; Section~\ref{sec:sec-analysis} presents a security analysis of our implementation; Section~\ref{sec:evaluation} provides preliminary performance results; and Section~\ref{sec:conclusion} concludes and highlight possible future research directions.

\section{Preliminaries}
\label{sec:preliminaries}
This section introduces the notation that will be used throughout this work and provides essential background information relevant to our study.
 
\subsection{Notation}
PPT is shorthand for probabilistic polynomial time Turing machine which refers to the class of algorithms that can be executed by a probabilistic Turing machine within a polynomial amount of time. 

By $x \leftarrow D$ we mean that $x$ is a random variable drawn from the distribution $D$. $A(.) \rightarrow x$ denotes that an algorithm A generates an output $x$ from an input $(.)$.

To denote integer numbers $\{1,2, \ldots, n\}$ we use the notation $\mathcal{N}_n$ and by $(s_i)_k$ we mean a sequence $s_1, \ldots, s_k$.

By m!, where m is a positive integer, we mean the product of all positive integers from m down to 1.

The notation || is utilized for two distinct purposes: the absolute value in the traditional sense, and indicating program size, referring to the number of elements or instructions within the program.

A function $\mu : N \rightarrow N$ is defined as negligible if its growth rate is slower than the inverse of any polynomial. In other words, for any positive polynomial $p$, there exists a natural number $n_0$ such that $\mu(n)$ is smaller than 1 divided by the value of the polynomial $p(n)$ for all $n$ greater than $n_0$.

To avoid introducing additional notation, we use $P$ to denote both the program $P$ itself and the set of its instructions.

We write $s_1$ for an instruction coming from the source program and $\hat s_1$ for the obfuscated version of $s_1$ generated by an obfuscator.

\subsection{Adversarial Model}
Adversaries are assumed to be honest but curious end users who possess physical access to the obfuscated program. They can run the program and monitor its internal states but are unable to tamper with it. While our end goal is to eventually generalize our techniques to adversaries able to tamper with the obfuscated program and execution environment, solutions to the latter will still need to resist passive observation so we focus on those passive capabilities first. Nonetheless, our adversary model may still be practical and relevant in scenarios where the program is stored in one-time writable memories that are resistant to tampering, such as a protected BIOS.

An efficient adversary, alternatively called a PPT adversary, is a computational entity that operates within a practical and feasible computational time frame, typically limited to polynomial time complexity, which corresponds to having access to a finite but arbitrarily large number of classical (non-quantum) computers to break the obfuscation.  Such an adversary may execute a polynomial number of individual attacks: If each attack has a negligible probability of success, then the adversary has a less-than-one probability of breaking the security scheme.  Increasing the size of the obfuscated programs linearly with random instructions provides an exponential increase in the difficulty of breaking the obfuscation.




\subsection{Background and Related Works}
\label{sec:background}

Software obfuscation is a technique that transforms computer programs into semantically equivalent versions that are difficult to understand \cite{collberg1997taxonomy}. It helps conceal the logic and purpose of the code, making it resistant to decompiling and reverse engineering. Additionally, it serves as a fundamental building block in cryptography enabling the construction of other cryptographic primitives, such as public key encryption, digital signature, multiparty secure computation, fully homomorphic encryption, etc. \cite{gonzalez2023software,sahai2014use}. 

The concept of obfuscation was initially introduced at the International Obfuscated C Code Contest in 1984, leading to the development of practical obfuscation heuristics. These heuristics often involve adding junk code and rearranging data or code segments to increase program complexity. However, they lack a formal definition for obfuscation and do not provide prove security guarantees. While examining different obfuscation heuristics, Xu et al. introduced a layered obfuscation approach as a practical framework to utilize these techniques \cite{xu2020layered}. They developed a taxonomy hierarchy that aids developers in selecting appropriate obfuscation techniques and designing layered obfuscation solutions tailored to their specific requirements.

Theoretical exploration of obfuscation began in~2001 with the influential work by Barak et al. \cite{barak2001possibility}, with the formalization of an ideal obfuscator, referred to as a "virtual black box (VBB)," through which an obfuscated program should only reveal its input-output functionality without disclosing any additional information. They showed that specially-constructed functions may reveal their internal state, regardless of the scheme used for obfuscation, therefore proving the impossibility of achieving VBB obfuscation in general for all possible programs. Nonetheless, the functions they construct \textit{are designed} to reveal their internal state through input-outputs, which is unlikely to be the case for practical programs that would benefit from obfuscation. This therefore leaves ample of room for designing techniques that do work on a practical subset of all possible programs.
On the other end of the theoretical spectrum, specific classes of functions, such as point functions \cite{canetti1997towards, lynn2004positive, wee2005obfuscating} or evasive functions \cite{zobernig2020mathematical}, have been shown to be obfuscatable in a black box manner. It has not been shown yet that \textit{only} those classes of functions can be VBB-obfuscated, which still leaves open the question of how large the class of VBB-obfuscatable functions is. Additionally, Goldwasser et al. demonstrated that achieving VBB obfuscation becomes more challenging, and in some cases impossible, when auxiliary input related to the adversary's prior information is available \cite{goldwasser2005impossibility}. They proved that there exist many natural classes of functions that cannot be obfuscated. 

In response to the negative results of general VBB obfuscation, Barak et al., proposed an alternative weaker notion of obfuscation called \textit{indistinguishability obfuscation (iO)} \cite{barak2001possibility}. An iO ensures that the obfuscation of any two programs with the same functionality cannot be distinguished from each other. It also can be considered as a fundamental low-level algorithmic operation that serves as a key building block in cryptographic systems \cite{gonzalez2023software}. The definition of iO does not provide an absolute guarantee against information leakage from the obfuscated program. However, it does offer a significant advantage by ensuring that alternative obfuscators cannot generate obfuscated programs with the same functionality that are more indistinguishable than those produced by iO \cite{goldwasser2007best}.

Goldwasser et al. conducted research on the concept of \textit{best possible obfuscation,} which allows an obfuscator to disclose some information that is also leaked or learnable by any alternative program \cite{goldwasser2007best}. In essence, a best possible obfuscator aims to minimize the amount of information it reveals while maintaining equivalent functionality. Goldwasser et al. also demonstrated that any efficiently computable indistinguishability obfuscator achieves this notion of best possible obfuscation \cite{goldwasser2007best}.

Garg et al. made significant progress in the field of indistinguishability obfuscation. They showed the existence of iO for all polynomial-size programs and introduced a candidate construction specifically designed for a certain complexity class of programs. To expand the applicability of iO to all polynomial-size programs, they utilized the capabilities of Fully Homomorphic Encryption (FHE) \cite{Garg2013PointFunction}.

The theoretical obfuscation approaches mentioned earlier, along with similar ones like Gray Box obfuscation \cite{kuzurin2007concept} and VBB for the random oracle model \cite{lynn2004positive}, have made attempts to address the lack of precise definitions and security guarantees in obfuscation. However, these approaches often rely on abstract models that are disconnected from real-world programs, posing challenges when it comes to practical implementation. Additionally, these approaches have a general understanding of obfuscation, independent of the particular purposes or original code of a program. For example, they exhibit the same behavior when obfuscating a program that contains secret keys as they do when obfuscating a program that incorporates a valuable algorithm.

In contrast to existing theoretical work, we approach obfuscation as a concrete program transformation technique, based on compilation and interpretation fundamentals. Rather than aiming for performance, as is the traditional focus of research in programming language implementations, we aim for sound obfuscation of every \textit{structural} and \textit{dynamic} aspects of programs, in a bottom-up systematic manner. In contrast to previous practical techniques based on heuristics, we ground our techniques with cryptographic definitions of the expected properties so that the obfuscation will resist stronger adversaries than usually implicitly assumed.


\section{Instruction Decorrelation}
\label{sec:obfISint}

Existing obfuscation heuristics, such as adding junk code or rearranging code sections, are effectively special cases of decorrelation in that they respectively lower the probability that any two instructions are correlated and make control-flow dependencies harder to identify. In this section, we generalize these heuristics to a principled decorrelation approach, starting with some basic definitions.

\subsection{Intuition}

As a preliminary motivation, we require an obfuscated program to still directly execute on regular hardware so that it can execute reasonably fast on a large number of hardware platforms. This implies that the instructions of the obfuscated and non-obfuscated versions of the same program will both ultimately be based on the same hardware instruction set. Since this approach precludes hiding the semantics and operations of individual instructions, instead, the instructions of \textit{multiple programs} are combined into a single obfuscated program so that it is difficult to associate the obfuscated instructions to their source programs, making the obfuscation a one-way transformation. Obfuscating a set of programs allows both combining multiple useful programs together \textit{and} making the inverse process harder by adding junk code, similar to existing heuristics.

 More precisely, by this perspective, an obfuscator is a transformation $\mathcal{O}$ that takes a set of programs $P=\{P_1, \ldots, P_n\}$ as inputs and interleaves their instructions to produce a new single obfuscated program, $\mathcal{O}(P)$. The obfuscated program should have three main properties, which we explain as follows.

First, and perhaps obviously, it should maintain the same functionality as running each individual program $P_1, \ldots, P_n$ separately. This property is commonly referred to as \textit{functionality preservation} \cite{barak2001possibility} and needs no additional discussion.

Second, informally the transformation should not slow down the execution so much that it changes the complexity class of the program. Formally, obfuscation should introduce only a polynomial slowdown factor in the running time of the obfuscated program so that a machine that terminates in polynomial time should still do after obfuscation~\cite{barak2001possibility}. This implies that the running time of $\mathcal{O}(P)$ should only be polynomially larger than the running time of all input programs. From an implementation point of view, this is a  weak requirement because, e.g., a 1000x slowdown is still polynomial.

Third, and this is our main contribution since the previous two requirements are fairly standard, \textit{instructions in each possible pairs in $\mathcal{O}(P)$ should appear independent from one another} in order to prevent adversaries from reconstructing input programs. Formally, given polynomial time for running static and dynamic analyses on $\mathcal{O}(P)$ and $P$ with a free choice of inputs, on a finite but arbitrarily large number of classical computers, an adversary $A$ should have no more than a negligible advantage for guessing that any two (different) obfuscated instructions originated  from the same input program $P_i$ than if it picked any two instructions at random. 


\subsection{Formal Definition}

Combining all three requirements described in the previous section, we formally define an \textit{instruction-independent obfuscation} as follows:
\begin{definition}[Instruction-Independent Obfuscation]
\label{RIObfs}
A probabilistic algorithm $\mathcal{O}$ is considered an instruction independent obfuscator for a set of input programs $P=\{P_1, \ldots, P_n\}$ if its output satisfies the following three conditions:
\begin{itemize}
	\item[--] \textbf{Functionality preservation}: $\mathcal{O}(P)$ computes the same functionality as running all individual machines in $P$.
	\item[--] \textbf{Polynomial slowdown}: There exist polynomials $p$ and $q$ such that $|\mathcal{O}(P)| \leq q(\sum_{i=1}^n |P_i|)$ and if $P_1, \ldots, P_n$ halts within a maximum $t$ of steps on inputs $x_1, \ldots, x_n$ respectively, the obfuscated machine $\mathcal{O}(P)$ halts within $q(t)$ steps on the same input.
	\item[--] \textbf{Unintelligibility}: for any PPT adversary A, and for all $i \in \{1, ..., n\}$ there exists a negligible function $\epsilon$ such that:	
	\begin{equation}
	\begin{split}
	\label{eq:unint}
		 Pr[ &A(\mathcal{O}(P),P) \rightarrow (\hat s_1,\hat s_2) \\
		& : \hat s_1, \hat s_2 \in \mathcal{O}(P) \wedge \hat s_1 \neq \hat s_2 \wedge s_1,s_2 \in P_i] \leq \\
		 Pr[ & \hat s_1, \hat s_2 \in \mathcal{O}(P) : \hat s_1 \neq \hat s_2 \wedge s_1,s_2 \in P_i] + \epsilon(.)
	\end{split}
	\end{equation}
\end{itemize}
\end{definition}

The probability that two different instructions of $\mathcal{O}(P)$ belong to the same program is obtained with the following equation:
	
\begin{align*}
\label{Eq:Qn}
&Pr[\hat s_1, \hat s_2 \in \mathcal{O}(P) : \hat s_1 \neq \hat s_2 \wedge s_1,s_2 \in P_i] =\\
&\frac{|P_1|}{\sum_{k=1}^n |P_k|}\times\frac{|P_1|-1}{(\sum_{k=1}^n |P_k|)-1} + \\
&\frac{|P_2|}{\sum_{k=1}^n |P_k|}\times\frac{|P_2|-1}{(\sum_{k=1}^n |P_k|)-1} + \ldots + \\
\displaybreak[0]
&\frac{|P_n|}{\sum_{k=1}^n |P_k|}\times\frac{|P_n|-1}{(\sum_{k=1}^n |P_k|)-1} \\
&= \frac{\sum_{k=1}^n |P_k|^2- \sum_{k=1}^n |P_k|}{(\sum_{k=1}^n |P_k|)^2- \sum_{k=1}^n |P_k|}
\end{align*}

In the special case where the length of programs is equal, i.e  $|P_i|=l$ for all $i \in \{1, \ldots, n\}$, the equation simplifies to $\frac{l-1}{nl-1}$ and as the programs' length becomes sufficiently large, the equation converges to $1/n$.

Obtaining a practical obfuscation scheme that satisfies Definition~\ref{RIObfs} is however not trivial: to achieve this, an obfuscation designer needs to carefully design together the program transformations and the execution environment on which the obfuscated programs will run. Moreover, Definition~\ref{RIObfs}  does not provide a criteria to verify that all transformations identified are \textit{sufficient}, it only states what resulting properties they should achieve together. Nonetheless, a number of useful additional properties can be deduced from it to obtain a principled methodology to obfuscate programs. In the next sections, we present those that our current implementation supports, later presented in Section~\ref{sec:construction}.

\subsection{Source Hiding}

The astute reader may have noticed that Definition~\ref{RIObfs} does not explicitly break the link between an individual instruction and the source program from which it originated, because it requires what may at first sight seem like a weaker property: namely, that two instructions from the same program may not be associated to one another, even if we do not know what that source program is. However, this property is stronger than it appears because it actually implies that an individual instruction may not be associated to the source program from which it originated:

\label{sec:instind}
\begin{theorem}
\label{th:singleinst} (source hiding):
Equation \ref{eq:unint} implies that an adversary cannot establish a relationship between any randomly selected instruction $s \in \mathcal{O}(P)$ and its originating program $P_i \in P$ with more than a negligible advantage compared to random guessing.
\end{theorem}
\begin{proof}
Suppose, for the sake of contradiction, that an adversary has the ability to relate a randomly selected instruction to its source program. In that case, the adversary can randomly choose at least $n+1$ instructions, where $n$ represents the number of input programs. According to the pigeonhole principle, there must be at least two instructions that belong to the same input program which means that the adversary can find at least two correlated instructions in polynomial time with probability 1. This contradicts Equation \ref{eq:unint}.
\end{proof}

\subsection{Uniformity of Instructions Distributions}

Our obfuscation approach preserves the basic operations performed by the source programs, e.g. addition and comparison, because they are ultimately implemented with the same hardware instructions. Adversaries may therefore try to leverage a priori knowledge about the distribution of instructions in the source programs: if they significantly differ, then observing certain instructions in the obfuscator output enables the adversary to associate these instructions back to the source programs. The obfuscator should therefore ensure that knowledge of the distribution of instructions of the programs provides no advantage. Formally:

\begin{theorem}
\label{th:distr-decorrelation} (source distribution uniformity):
Theorem \ref{th:singleinst} implies that input programs must have identical probability distributions and almost the same size.
\end{theorem}
\begin{proof}

(sketch) 
Without prior knowledge, the adversary may only randomly guess the provenance of a randomly selected instruction with a probability of $1/n$. Having access to the probability distributions of input instructions for each program, the adversary makes educated guesses with a probability improved by $\frac{n-1}{n^2}$ compared to random guessing, providing a non-negligible advantage because it is a polynomial fraction. This contradicts Theorem \ref{th:singleinst}. The full proof can be found in Section \ref{app:proof1} of the appendix.
\end{proof}

The simplest approach we have found so far of ensuring Theorem~\ref{th:distr-decorrelation} is to simply add meaningless instructions to every input programs so that  they all follow a uniform distribution across all possible hardware instructions. 



 
\subsection{Decorrelation of Control Flow}

Adversaries may also try to leverage control flow dependencies between instructions of the obfuscated program, e.g., when a comparison sets internal processor registers that affects the branching behaviour of a following instruction. When observing obfuscated instructions that still depend on one another because of control-flow, the adversary may correctly conclude that they originated from the same input program. Generalizing, the obfuscator should decorrelate control-flow dependencies of the obfuscated program from those of the input programs. Formally:

\begin{theorem}
\label{th:cfg-decorrelation} (control-flow decorrelation):
Equation \ref{eq:unint} implies either or both following conditions: 
\begin{enumerate}
     \item obfuscated instructions that originated from the same input program, i.e. $\hat s_1, \hat s_2 \in \mathcal{O}(P) :  s_1, s_2  \in P_i$, should not have control-flow dependencies in the obfuscated program $\mathcal{O}(P)$; 
     \item the obfuscator should introduce random control-flow dependencies between instructions not coming from the same input program, $\hat s_1, \hat s_2 \in \mathcal{O}(P) : s_1, s_2  \notin P_i$.
\end{enumerate}
\end{theorem}
\begin{proof}
Suppose, for the sake of contradiction, that the obfuscator does introduce control-flow dependencies between any two instructions $\hat s_1, \hat s_2 \in \mathcal{O}(P) :s_1,s_2 \in P_i $ and does not introduce control-flow dependencies between in two instructions $\hat s_3, \hat s_4 \in \mathcal{O}(P) : s_3 \in P_i \wedge s_4 \in P_j \wedge i \neq j $. Therefore, when the adversary sees any two instruction $\hat s_1, \hat s_2 \in \mathcal{O}(M)$ with control-flow dependencies he knows with probability 1 that $s_1, s_2 \in M_i$, which contradicts Equation \ref{eq:unint}.
\end{proof}

Quantifying how many random control-flow dependencies should be added to obfuscate the original dependencies requires a probability analysis. For the rest of this paper we take the simpler position of removing any control-flow dependencies from any two obfuscated instructions. This has the added benefit of making it easier to combine multiple obfuscated programs together.

%
%
 
\subsection{Decorrelation of Data Dependencies}
\label{sec:Chinput}

If two instructions read or write to the same variable using the same label or the same memory offset, they most likely belong to the same program, otherwise this may violate functionality preservation. Adversaries may leverage this information to correlate instructions. Because labels may be statically defined and corresponding memory locations may be dynamically computed during execution, decorrelating instructions requires both a static and dynamic obfuscation of data dependencies. Formally:

\begin{theorem}
\label{th:data-access-decorrelation} (data dependencies decorrelation):
Equation \ref{eq:unint} implies either or both of the followings:
\begin{enumerate}
	\item accesses to the same label or memory location by obfuscated instructions that originated from the same input program, i.e. $\hat s_1, \hat s_2 \in \mathcal{O}(P) :  s_1, s_2  \in P_i$ should appear as accesses to different labels or memory locations to an adversary; 
	\item some obfuscated instructions that originated from different input programs, i.e. $\hat s_1, \hat s_2 \in \mathcal{O}(P) :  s_1  \in P_i \wedge s_2 \in P_j \wedge i \neq j$ should access the same label or memory location.
\end{enumerate}
\end{theorem}
\begin{proof}
Suppose, for the sake of contradiction, that accesses to the same label or memory location from instructions that originated from the same input program $\hat s_1, \hat s_2 \in \mathcal{O}(P) :  s_1, s_2  \in P_i$ still appear to the adversary as accesses to the same label or memory location, either statically or during execution, and that obfuscated instructions that originated from different input machines never access the same label or memory location. Therefore, if an adversary observes two obfuscated instructions  $\hat s_1, \hat s_2 \in \mathcal{O}(P)$ that access the same label or memory location, either statically or during execution, he knows with probability 1 that $s_1, s_2 \in P_i$, which contradicts Equation \ref{eq:unint}.
\end{proof}

Theorem~\ref{th:data-access-decorrelation} implies that instruction decorrelation requires both a static program transformation that rewrites statically-defined memory accesses and a run-time component that makes sequential memory accesses to the same location also use what appears to the adversary as different offsets. Of the three properties we have derived, this is the most costly and complex to implement, and the only one to require a runtime component. With this conceptual background, we now present how a practical obfuscator can meet those properties.

\section {Implementation}
\label{sec:construction}
In this section, we present an implementation of an obfuscator that hides which source program an obfuscated instruction originated from (Theorem~\ref{th:singleinst}) and  data access dependencies between instructions (Theorem~\ref{th:data-access-decorrelation}). 

Our implementation expects input programs that already have similar instructions distributions (Theorem~\ref{th:distr-decorrelation}) and only use the instructions of a custom intermediate representation that are control-flow independent (Theorem~\ref{th:cfg-decorrelation}). Both of these properties are relatively straight-forward to achieve using well-known compilation techniques, so we only describe at a high-level how to achieve them and we focus most of our discussion on how to achieve source and data dependency hiding.

The rest of this section is structured as follows. Section~\ref{sec:implementation:overview} provides an overview of the architecture of our implementation. Section~\ref{sec:implementation:hiding-cfg-deps} describes how to hide control-flow dependencies. Section~\ref{sec:implementation:uniform-id} describes how the distribution of instructions for input programs can be made uniform. Section~\ref{sec:implementation:hiding-source-program} explains how the obfuscator hides the source program of instructions. Section~\ref{sec:implementation:-hiding-data-dependencies} discusses the method employed by the obfuscator to hide data dependencies within programs instructions.  Section~\ref{sec:example} illustrates the behaviour of our implementation with a compilation example.

\subsection{Architecture Overview}
\label{sec:implementation:overview}

\begin{figure}[th]
    \centering
    \includegraphics[width=0.5\textwidth]{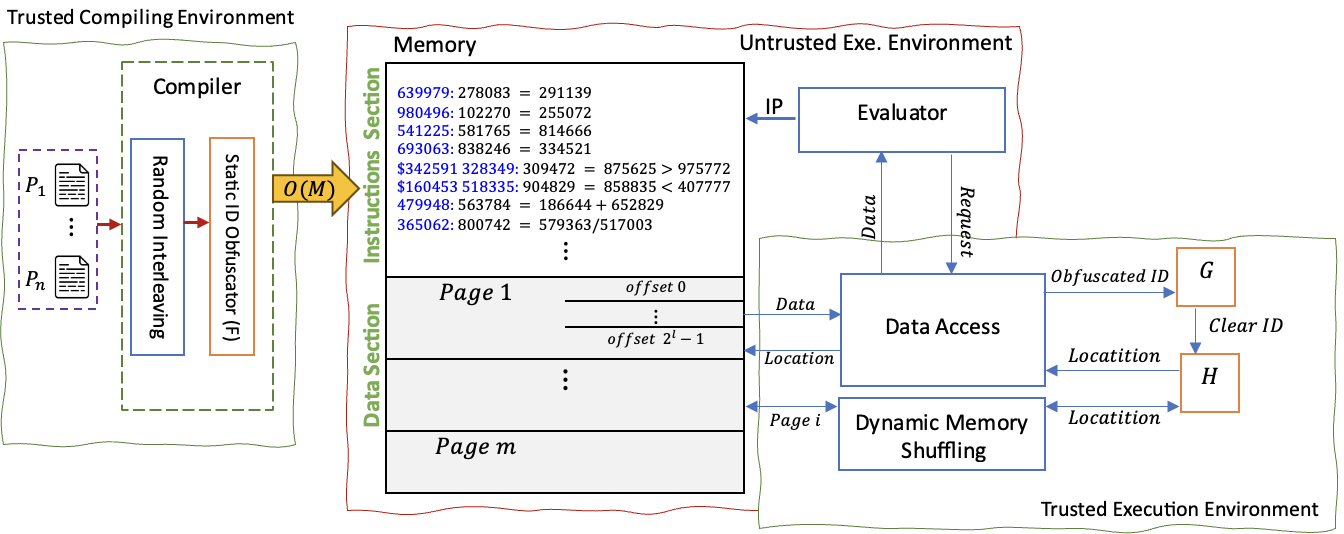}
    \caption{Architecture of our implementation}
    \label{fig:overview}
\end{figure}

Figure~\ref{fig:overview} provides an overview of our implementation. Prior to execution, a compiler, shown on the left side of the figure, translates input programs into instructions in the format expected by our interpreter, shown on the right. This compiler performs two static transformations: 1) it hides the original source of instructions by randomly interleaving the instructions of multiple input programs; 2) it removes data dependencies from the instructions by obfuscating every label used to access memory locations using a secret key $sk$ to make them appear as random numbers. The compiler executes within a trusted execution environment (TEE) to hide $sk$ from the adversary. This secret key is later securely shared with the interpreter to keep it hidden from the adversary.

The output of the compiler is stored in a read-only part of memory, outside of the trusted execution environment, shown in the middle of Figure~\ref{fig:overview} as the instruction section. The rest of the memory, which we call the data section, may be used for reading and writing data by the program. That data section is periodically shuffled to break the correlation of sequential data accesses of related instructions.

Instructions are executed by the evaluator. The next instruction to execute is stored in an \textit{instruction pointer} (IP) within the evaluator, outside of the TEE. The evaluator is however only one part of the interpreter: another part resolves memory locations within a TEE because the translation back to memory locations requires the secret key $sk$.

Finally, periodic memory shuffling is performed by reading individual blocks of the main memory, also called \textit{pages}, and shuffling data within a block inside the TEE. This shuffling is performed every $n$ data access request on average so that an adversary may not find correlations by simply querying for all labels to learn the current mapping. The functions $G$ and $H$ used to retrieve the correct memory locations are described later in Section~\ref{sec:implementation:-hiding-data-dependencies}.

\subsection{Hiding Control-Flow Dependencies within a Single Program}
\label{sec:implementation:hiding-cfg-deps}
\label{sec:predicate}

Our implementation expects instructions that are not correlated due to control-flow dependencies, e.g., a comparison instruction that sets an internal register of the CPU followed by a branching instruction that is conditional on the state of that internal register are correlated. Control-flow independence is enforced through the use of an intermediate representation that reifies those dependencies as \textit{predicates} whose value is computed eagerly and stored in memory as any other data item. This way, the techniques used to decorrelate data dependencies, described later in Section~\ref{sec:implementation:-hiding-data-dependencies}, can also be used to hide the control-flow dependencies as well.

\subsubsection{Predicate-based Intermediate Representation}

\begin{table}
  \centering
  \caption{Concrete Syntax of $L_{cfi}$ Language}
  \label{tbl:syntax}
  \begin{tabular}{|ll |cc|}
    \hline
    type ::= & bool \:|\: int \:|\: float \\
    declar ::= & type \: var \\
    cmp ::=  & $ < \:|\: \leq \:|\: > \:|\: \geq \:|\: == \:|\: !=$ \\
    bools ::= & true \:|\: false \\
    exp ::= & bools\:|\: int \:|\: bool\_var\:|\: int\_var \:|\: \\
      & $ -exp  \:|\:exp + exp \:|\: exp-exp \:|\:$ \\
    & $exp *exp \:|\:exp / exp\:|\:exp \% exp \:|\:$ \\
     & $\sim exp\:|\:exp || exp \:|\: exp\: \&\& \:exp \:|\:$ \\
     & $exp \:cmp\: exp\:|\: (exp) $ \\
     pred ::= & bools \:|\: bool\_var \\
     inst ::= & print("string", exp) \:|\: exp \\
     stmt ::= & (\$label)? \:|\: pred : inst \\
     $L_{cfi} ::=$ & $declar^* \: stmt^*$ \\
    \hline
  \end{tabular}
\end{table}

Syntactically, we write statements with predicates as \textit{predicate:instruction}. The predicate of an instruction contains the conjunction of all conditions that must be true for that instruction to be executed. For example, this may include conditions tested in \texttt{if} statements. Table \ref{tbl:syntax} shows the concrete syntax of the language $L_{cfi}$, which is used for writing all input programs $P_1, \ldots, P_n$.

\begin{figure}[ht]
  \begin{subfigure}[b]{1.2 in}
    \includegraphics[width=.9\textwidth]{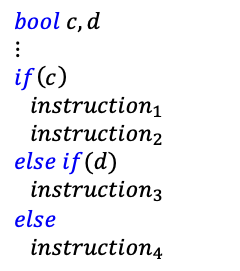}
    \caption{Regular}
    \label{fig:regular}
  \end{subfigure}
  \begin{subfigure}[b]{1.2 in}
    \includegraphics[width=.9\textwidth]{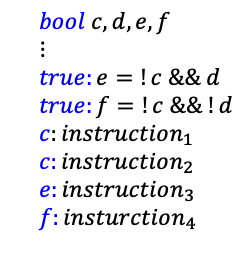}
    \caption{Predicate based}
    \label{fig:redundant}
  \end{subfigure}
  \caption{Control-Dependent vs. Control-Independent Syntaxes}
   \label{fig:regredundant}
\end{figure}

Figure~\ref{fig:regredundant} illustrates how an if statement can be translated to a predicate form. The semantics of the if statement implies that $instruction_1$ and $instruction_2$ should only be executed if condition $c$ is true,  $instruction_3$ is only executed if $c$ is false and $d$ is true, and $instruction_4$ should only be executed if both $c$ and $d$ are false. This is represented in predicate form with explicit variables and boolean operations, including the $"!"$ characte to represent logical negation and $"\&\&"$ is used to denote logical conjunction (logical AND). 

In addition to decorrelating instructions from control-flow dependencies, this approach also simplifies interleaving the instructions of multiple programs, which we will discuss in more detail in Section~\ref{sec:implementation:hiding-source-program}.

\subsection{Uniformizing Instruction Distributions}
\label{sec:implementation:uniform-id}

As mentioned in section \ref{sec:obfISint}, we assume that the adversary has prior knowledge of the probability distribution of input programs. Figure \ref{fig:dist} illustrates this information for programs $P_1$ and $P_2$, with the corresponding code snippets available in Figure \ref{fig:example}. By examining these distributions, the adversary can observe that the probability of encountering a division instruction in $P_1$ is zero. Therefore, when they come across a division instruction in $\mathcal{O}(P)$, they can deduce that it belongs to $P_2$. Additionally, the adversary gains an advantage with respect to addition instructions, as they are more likely to be associated with the first program.

\begin{figure}[H]
  \begin{subfigure}[b]{1.2 in}
    \includegraphics[width=1.3\textwidth]{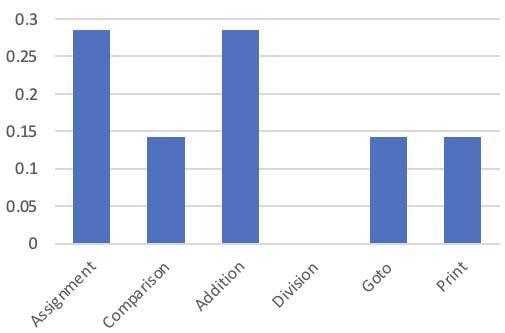}
    \caption{$P_1$ Distribution}
    \label{fig:exInterleaved}
  \end{subfigure}
  \begin{subfigure}[b]{1.2 in}
    \includegraphics[width=1.3\textwidth]{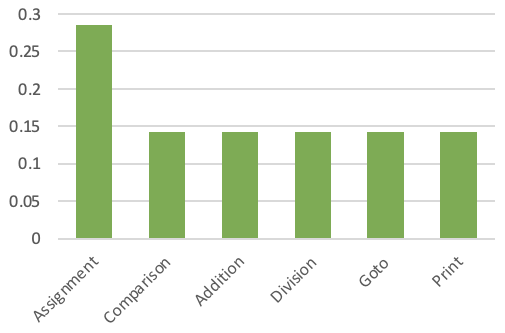}
    \caption{$P_2$ Distribution}
    \label{fig:exObfuscated}
  \end{subfigure}
  \caption{Probability distribution of $P_1, P_2$}
 \label{fig:dist}
\end{figure}

\subsection{Hiding the Source Program}
\label{sec:implementation:hiding-source-program}




To preserve data dependencies among instructions, we ensure the relative order of instructions is maintained across all programs. This is achieved through the random interleaving process. As Algorithm  \ref{alg:interleaving} shows, the process in each step begins by randomly selecting a program from the set of all programs ${P}=\{P_1, \ldots, P_n\}$, where the selection probabilities are based on the size of each program. Then, it selects the last unseen instruction within the chosen program to append it to the intermediate result. 

\begin{algorithm}[h]
\caption{Random Interleaving}
\label{alg:interleaving}
\begin{algorithmic}[]
\Statex

\State $result:=[]$
\State $count:=\sum_{i=1}^n |P_i|$

\While {$count \geq 0$}
	\State $candidate:=-1$
	\State $r \leftarrow randomBetween(0,count-1)$
	\For { $i:=1$ to $n$}
	       	\If{$r<|P_i|$} 
    		\State $candidate := i$
		\State \textbf{break}
		\EndIf 
	\EndFor
	\If{$candidate <0$}
	 \State $candidate := n$
	\EndIf
	\State $s:= \Call{Pop}{$$P_{candidate}$$}$
	\State $\Call{Push}{$$result,s$$}$
	\State $count = count-1$
\EndWhile
\State \Return result
\end{algorithmic}
\end{algorithm}

\subsubsection{Instruction Independence within CFG Cycles}

 Interleaving instructions from distinct loops can introduce complexities and potential issues. Even in a seemingly straightforward scenario where one loop is interleaved with sequential instructions from another program, unintended execution of certain instructions may occur. When the interleaved loop redirects execution back to its own instructions, instructions from the second program with predicates evaluating as true may mistakenly be executed.
To mitigate this issue, we first introduce a mechanism that prevents the re-execution of previously executed instructions, and then utilize this mechanism to prevent loop misbehaviors. 

The re-execution prevention mechanism is accomplished by augmenting our predicate structure with an additional field, initially set to -1, to store the last execution line in which the predicate was evaluated. During the execution of the interleaved program $\mathcal{O}(P)$, for each instruction, the evaluator compares this field with the current line number. If the stored line number is greater than the current line number, the evaluator skips the execution of the instruction. Algorithm \ref{alg:predEval} demonstrates the evaluation process for predicates. This mechanism ensures that once an instruction is executed, the field \textit{last line index} of its predicate is updated. As a result, the instruction will not be executed again because its predicate evaluation will return false.
 
\begin{algorithm}[H]
\caption{Predicate Evaluation}
\label{alg:predEval}
\begin{algorithmic}[]
\Statex

\If{$CurrentLineIndex()<predicate.LastLineIndex$} \State \Return $false$ \EndIf 
\State $predicate.LastLineIndex = CurrentLineIndex() $
\State \Return $predicate.value$
\end{algorithmic}
\end{algorithm}

To prevent loop misbehaviors, each branching instruction instructs the evaluator to reset the field \textit{last line index} of loop instructions, i.e set them to -1. By doing so, the loop instructions can be re-executed while preventing the execution of other instructions due to this re-execution prevention mechanism. The List of all predicates to be rest can be either recognized by the compiler itself or provided by the programmer explicitly.  

\begin{figure}[ht]
    \includegraphics[width=0.35\textwidth]{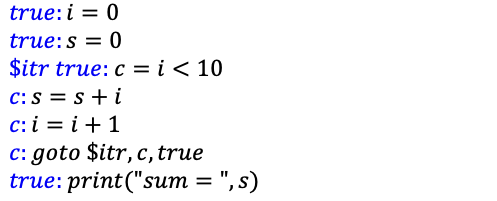}
    \caption{Simple loop example }
    \label{fig:loopexample}
\end{figure}

Figure \ref{fig:loopexample} provides a concrete example illustrating the summation of numbers from 0 to 9 using a loop. The first argument of the goto instruction specifies the line label for controlling the flow of execution, while the remaining two arguments represent lists of predicates that require resetting during the execution process. As discussed in Section \ref{sec:implementation:-hiding-data-dependencies}, in our construction, constant values are treated as variables. Therefore the constant value "true" in Figure \ref{fig:loopexample} is also included in the list of predicates that should be reset.

\subsection{Hiding Data Dependencies}
\label{sec:implementation:-hiding-data-dependencies}


In order to conceal data dependencies among instructions, we employ both static and dynamic mechanisms. The static mechanism assigns random labels to data items during compile time, while the dynamic mechanism periodically shuffles the memory data segment.

\subsubsection{Replacing Memory Labels with Random IDs}
\label{sec:randomOffsets}

Having randomly interleaved instructions, to obscure the data linkage between instructions we employ non-repetitive random labels for each memory reference in any instruction. To achieve this we organize data section of the obfuscated program in a flat structure where all constant values and variables from all input programs are stored. The first byte of each data item in this flat architecture is used as a reference to the corresponding data item. Figure \ref{fig:flatstructure} illustrates this concept, where integer variables are represented as 4 bytes and Boolean variables as 1 byte.

All instructions in the obfuscated program utilize these labels or IDs to refer to the data elements. However, if we were to directly use these clear IDs in the compiled output, adversaries could still discern the connections between instructions. Therefore, we employ a mechanism that randomizes these labels, making them unpredictable to potential adversaries. 

\begin{figure}[ht]
    \centering
    \includegraphics[width=0.2\textwidth, height=0.2\textheight]{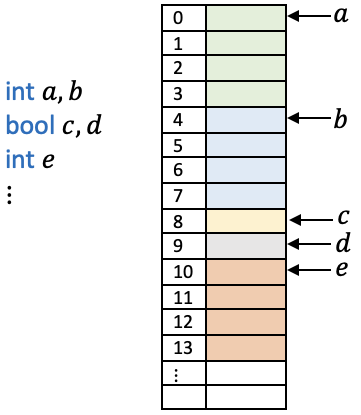}
    \caption{Flat Memory Structure of Data Section}
    \label{fig:flatstructure}
\end{figure}

In essence, for every label $l$, we consider a set of random numbers $R$. When the compiler encounters a reference to $l$ in any instruction, it selects one element from $R$ and replaces it with $l$. Conversely, in the runtime environment, when the evaluator needs to resolve a memory reference, it first identifies the set $R$ associated with each data item and then determines the clear ID. The two following functions formalize this mechanism:

\begin{enumerate}
 \item $F(sk, x) = R$
\item $G(sk, r) = x$
\end{enumerate}

where $x$ is a data ID, $sk$ is a secret key and $R$ is a collection of random integers, each serving as an obfuscated ID.

In our implementations, we construct a simple set of functions $F$ and $G$ utilizing modular arithmetic. The $sk$ is a random numbers in $[\alpha t, \beta t]$ where $t$ is the size of flat memory structure and $\alpha$ and $\beta$ are two integer parameters such that $\alpha < \beta$. The function $F$, for each clear ID $x$, mines a random number from a bounded range, like all 6 digit integers, that belongs to a congruence class modulo $sk$. Conversely, the function $G$ easily retrieve the ID $x$ by calculating the remainder of the obfuscated ID modulo $sk$. While this construction appears simple and may not be secure enough for certain applications, it is possible to design and employ more robust functions that satisfy the aforementioned properties.

\begin{figure}[H]
  \begin{subfigure}[b]{1.2 in}
    \includegraphics[width=.9\textwidth]{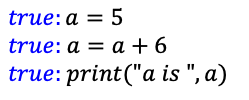}
    \caption{Normal}
    \label{fig:normal}
  \end{subfigure}
  \begin{subfigure}[b]{1.2 in}
    \includegraphics[width=1.42\textwidth]{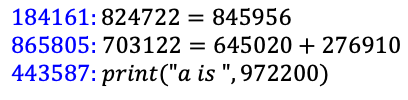}
    \caption{ID obfuscated}
    \label{fig:offsetobfs}
  \end{subfigure}
  \caption{Static ID obfuscation using modular arithmetic scheme with $sk=38$. }
 \label{fig:addrobfs}
\end{figure}

The code snippet in Figure \ref{fig:addrobfs} demonstrates a straightforward example along with its static ID obfuscated version, implemented using a simple modular arithmetic scheme. In the obfuscated version, the memory location associated with the variable $a$ is accessed using a unique random number for each occurrence of $a$. Similarly, the literal "true" is treated as a variable holding the value "true" and is referenced using different random numbers, similar to the variable $a$.

\subsubsection{Randomizing Memory Content}
\label{sec:shuffling}

While static ID obfuscation mechanism can provide resistance against \textit{static analysis}, it is still possible for adversaries to monitor the memory during the execution of the obfuscated program and observe which instructions access the same memory locations. In order to thwart this type of \textit{dynamic analysis}, we periodically modify the location of the data associated with each statically obfuscated ID such that it will appear random and unpredictable to any potential adversary. 

To achieve this, we employ a pseudo-random permutation function $H:\{0,1\}^k \times \{0,1\}^l \rightarrow \{0,1\}^l$ within the running environment. This function is employed to periodically shuffle each $2^l$-byte page of our flat memory structure, utilizing a k-bit counter. More precisely, by taking into account the value of the counter, the function $H$ maps a unique new offset to each byte based on its clear ID. Subsequently, the running environment shuffles the entire page according to these new offsets.

To gain a visual understanding of the function $H$, it can be envisioned as a table comprising $2^k$ columns. Each column represents a random permutation of numbers ranging from 0 to $2^{l-1}$. Additionally, a counter is utilized to indicate a specific column. Figure \ref{fig:shuffling} offers a graphical representation that illustrates this perspective.

\begin{figure}[ht]
    \centering
    \includegraphics[width=0.40\textwidth]{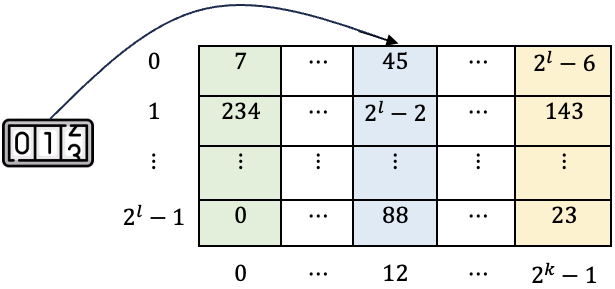}
    \caption{Dynamic offset obfuscation}
    \label{fig:shuffling}
\end{figure}




\subsection{Obfuscation Example}
\label{sec:example}
In this section, we show the application of our obfuscator to a practical example. To demonstrate the capabilities of the system, we utilize a pair of simple programs written in our control-independent programming syntax, represented as the input program set $P=\{P_1, P_2\}$. The program $P_1$ calculates the sum of numbers from 0 to 9, while $P_2$ computes the sum of a series of numbers $2, 4, 8, \ldots, 1024$. Figure \ref{fig:example} provides the source code for these two programs.

\begin{figure}[ht]
  \begin{subfigure}[b]{1.2 in}
    \includegraphics[width=1.0\textwidth]{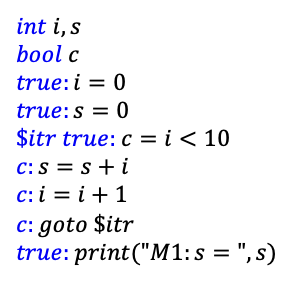}
    \caption{$P_1$ source code}
  \end{subfigure}
  \begin{subfigure}[b]{1.2 in}
    \includegraphics[width=1.0\textwidth]{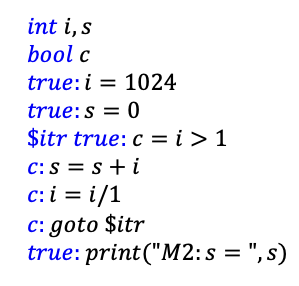}
    \caption{$P_2$ source code}
  \end{subfigure}
  \caption{Concrete Example of Two Programs }
 \label{fig:example}
\end{figure}

\subsubsection{Compilation}

Figure \ref{fig:exInterleaved} shows a randomly interleaved version of the two programs. In our implementations, we use the postfix "\_i" for each variable or label that belongs to program $P_i$. It is important to note that in this listing the constant value \textit{true} is used as a predicate, and since predicates have a \textit{last line index} field, each program has its own constant value \textit{true}, while other constant values could be considered common among all programs.

Furthermore, the output presented in Figure \ref{fig:exInterleaved} is an intermediate result produced by the compiler within a trusted environment, inaccessible to adversaries.  Figure \ref{fig:exObfuscated} illustrates the compiler output accessible to adversaries, which has undergone \textit{static ID obfuscation} using the key $sk=109$. It is important to note that in this implementation, strings are presented in their plain format, which allows an adversary to easily recognize their relevance to the original programs. To address this issue, cryptographic schemes or other obfuscation techniques can be employed. These additional measures can be implemented without affecting the overall integrity of our general model.

\begin{figure}[ht]
  \begin{subfigure}[b]{1.25 in}
    \includegraphics[height=1.25in]{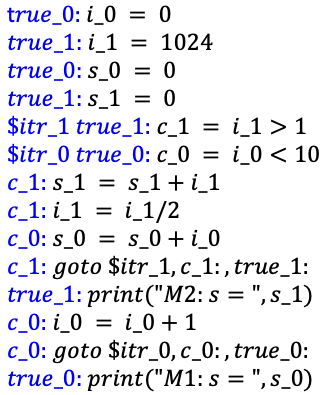}
    \caption{Interleaved $P_1$, $P_2$}
    \label{fig:exInterleaved}
  \end{subfigure}
  \begin{subfigure}[b]{1.25 in}
    \includegraphics[height=1.25in]{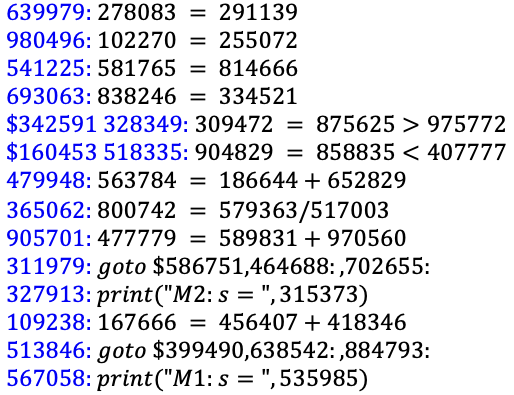}
    \caption{ID obfuscated}
    \label{fig:exObfuscated}
  \end{subfigure}
  \caption{Compiler output for two input programs $P_1, P_2$}
 \label{fig:example}
\end{figure}

\subsubsection{Runtime Environment Behavior}
To illustrate how the runtime environment simultaneously manages static and dynamic memory obfuscations, consider the value "278083" in Figure \ref{fig:exObfuscated} that corresponds to the 4-byte variable "$i\_0$". By utilizing the function $G$, introduced in section \ref{sec:randomOffsets}, the running environment can calculate $G(109, 278083)=278083\% 109 = 24$ to determine the clear ID of the variable "$i\_0$". Since "$i\_0$" is a 4-byte variable, the labels 24, 25, 26, and 27 would refer to its individual bytes. With this knowledge and the current value of the circular counter, the running environment utilizes the function $H$ to compute the dynamic addresses of these bytes. Specifically, when the counter value is 12, the running environment can compute $H(12,24)$, $H(12,25)$, $H(12,26)$, and $H(12,27)$ to determine the corresponding dynamic addresses. By utilizing these addresses, the running environment can successfully retrieve the value of the variable "$i\_0$".

It is crucial to emphasize that, in our constructions of the functions $F$, $G$, and $H$, must be invoked within a \textit{trusted execution environment} to prevent adversaries from accessing the calculation process.

\section{Security Analysis}
\label{sec:sec-analysis}

For a strong obfuscation construction in which all instructions look independent to PPT adversaries we can formalize the probability that an adversary can reconstruct a target program $P_t$. Suppose that all instructions in $P_t$ are sequentially dependent to each other, i.e each instruction $s_{i+1}$ can be executed only after execution of instruction $s_i$ for all $i\in \mathcal{N}_{|M_t|}$. In this case, among all possible permutations of instructions only one of them implements the correct algorithm of $P_t$. This means that among all $(\sum_{k=1}^n|P_k|)!$ possible permutations of $\mathcal{O}(P)$, only $((\sum_{k=1}^n|P_k|)-|P_t|)!$ possible permutations reconstruct the program $P_t$. In other words, the probability of success for an adversary in this case is as follows:
\begin{equation}
\label{eq:win}
Pr[A\ wins]=\frac{(\sum_{k=1}^n|P_k|)-|P_t|)!}{(\sum_{k=1}^n|P_k|)!}
\end{equation}
In the special case of equal-size input programs, i.e $|P_i|=l$ for all $i$s, the Equation \ref{eq:win} simplifies as the following negligible function:

\begin{equation}
\label{eq:winEqal}
\begin{split}
Pr[A\ wins]=\frac{(nl-l)!}{(nl)!}=\frac{1}{(nl)(nl-1) \ldots (nl-l+1)}\\
<\frac{1}{(nl-l)(nl-l)\ldots(nl-l)}=\left(\frac{1}{(n-1)l}\right)^l
\end{split}
\end{equation}

Equation \ref{eq:winEqal} demonstrates that the security of our approach relies on two factors, the number of input programs and the length of each program.

In our construction, the security of the output generated by the compiler relies on various factors including selecting the appropriate set of input programs, designing the functions $F, G$, and $H$ effectively, and ensuring a trusted execution environment for these functions.

\textbf{Design of functions}: In our construction, we have introduced functions $F$ and $G$ in section \ref{sec:randomOffsets}, as well as function $H$ in section \ref{sec:shuffling}, which offer certain security guarantees. These functions have been incorporated into our abstract design in a manner that allows for flexibility in their choice and implementation. Naturally, the selection and design of a specific approach depend on the desired security requirements and the associated cost and benefit considerations.

In this paper, we have implemented functions $F$ and $G$ using a simple design based on modular arithmetic. However, it is worth noting that these functions can also be designed using cryptographic primitives to enhance their strength and security depending on specific needs and goals of the system.

\textbf{Trusted execution environment}: In the current design, all three functions $F$, $G$, and $H$, incorporate the use of secret keys. The function $F$ is implemented in the compiler, which operates within a trusted execution environment according to our assumptions. On the other hand, the functions $G$ and $H$ are invoked in the evaluator, which is subject to monitoring by adversaries. To prevent the disclosure of secret keys, these invocations must be executed within a trusted execution environment.

In addition, it is crucial that dynamic memory shuffling, as outlined in section \ref{sec:shuffling}, takes place within a trusted execution environment. If performed outside of this secure environment, adversaries may be able to monitor and disrupt the shuffling process, rendering it ineffective. Therefore, it is necessary to ensure that the shuffling is conducted in a trusted execution environment. Once the randomized memory pages have been generated, they can then be safely transferred to the memory of an untrusted execution environment.

Considering the input programs $P_1, \ldots, P_n$, where their lengths are sufficiently large and almost equal, the probability of a randomly selected instruction belonging to a particular program is approximately $1/n$. This implies that, on average, and based on the pigeonhole principle, within a sequence of $n+1$ consecutive instructions in $\mathcal{O}(P)$, two of them belong to the same input program. Therefore, if we need strong security guarantees it is necessary to shuffle the modified memory pages after the execution of every $n$ consecutive instructions. This process, while ensuring stronger security, can impose a computational burden on our model. As a result, for different applications, it is essential to establish a tradeoff between this overhead and the desired level of security guarantees. By carefully considering the specific requirements and constraints of each application, we can determine the appropriate compromise between computational efficiency and the level of security needed.

\section{First Implementation Results}
\label{sec:evaluation}

We wrote an interpreter for a simple CPU featuring basic
  arithmetic operations, comparison and jump instructions for which
  our compiler generates obfuscated code. Both the instructions and
  the memory content are in cleartext format and the interpreter
  (potentially under control of the attacker) simply executes the
  obfuscated binary. We also implemented the trusted memory mapping
  and reshuffeling unit through which all memory access (to the
  cleartext content) has to pass.

In order to get a first ballpark figure on the runtime overhead
  introduced by such an interpretation, we wrote two small programs
  that were compiled into one obfuscated binary. The first program~P1
  calculates the average of an integer list of 4000 numbers, while the
  second program~P2 computes the dot product of two 10000 element
  vectors. As shown in Table~\ref{tbl:expres}, the runtime of the
  combined and obfuscated program only increases by roughly 10\% when
  compared to the sum of the interpretation runtimes of the original
  two programs. Note that for this performance test we did not
  equalize the instruction probabilities nor the loop durations. As a
  comparison point we also measured that the same functionality
  compiled to native execution on the Apple M2 chip terminates after
  1~millisec while in Python it takes 3 milllisec. Clearly, interpretation comes with a huge penalty, but
  the obfuscating transformation itself does not add significantly to
  the overall execution time.

\begin{table}[htbp]
  \centering
  \caption{Independent runtimes of two input programs P1 and P2 compared to the
    runtime of the obfuscated merged program $\mathcal{O}(\{P1,P2\})$, in seconds and
    averaged over 10 test runs.}
  \label{tbl:expres}
  \begin{tabular}{|ccc|c|c|}
    \hline
    P1 & P2 & Sum &  $\mathcal{O}(\{P1,P2\})$ & Overhead \\
    \hline
    0.6655 & 2.0616 & 2.7271 & 2.9772 & 9.17\% \\
    \hline
  \end{tabular}
\end{table}


\section{Conclusions and Open Problems}
\label{sec:conclusion}

In this paper we explore practical ways of obfuscating programs
  based on randomized instruction interleaving instead of cryptographic program
  transformation. By combining several programs to obfuscate into
  one big program to execute, a compiler can impose statistic
  properties of the resulting code that makes it impossible for
  an attacker to link single instructions to one of the original
  input programs. Moreover, by memory address rewriting and
  reshuffling of memory cells at run time, we reduce the
  elements of a CPU that must execute inside a trusted environment
  while preventing dynamic analysis of the running code.
  Our definition of unintelligibility focuses on the
  independence of any two instructions, from which we derive
  several desirable properties that the obfuscated code
  must have. We wrote a first compiler that is able to enforce
  some of these properties, most importantly the decorrelation
  of an instruction from its origin, but also the decorrelation
  of memory access to any of the original input programs.

Our current approach assumes an honest-but-curious attacker
  which does not tamper with neither the obfuscated instructions
  nor the memory content. In a next step we want to address
  a man-at-the-end scenario where an attacker can actively
  perturbate the program execution in order to extract
  data and instruction dependencies. A second research
  challenge is to reduce the footprint of the currently
  required trusted unit for memory address
  translation and shuffling. Removing the need of a TEE could make our approach more widely applicable. 


\begin{acks}
We would like to express our sincere gratitude to Osman Bicer for numerous insightful discussions concerning the security properties of this work.
\end{acks}

\bibliographystyle{ACM-Reference-Format}
\bibliography{main}

\appendix
\section{Proof for Theorem \ref{th:distr-decorrelation}}
\label{app:proof1}
The total number of instructions of type s drawn from program $P_i$ is equal to $|P_i|*D_i(S=s)$. Therefore, the probability that a random instruction which is drawn from $\mathcal{O}(P)$ originates from $P_i$ is as follows:

\begin{equation}
\label{eq:totalnum}
Pr[s \in P_i]=\frac{|P_i|D_i(S=s)}{\sum^n_{k=1}|P_k|D_k(S=s)}
\end{equation}

On the other hand, the probability that an adversary wins the game of correctly guessing the origination of an instruction is:

\begin{equation}
\label{eq:genwin}
Pr[win]=\sum^n_{i=1}Pr[selects \ P_i]*Pr[s \in P_i]
\end{equation}

Putting \ref{eq:totalnum} to \ref{eq:genwin} will give:
\begin{equation}
\label{eq:maingame}
\begin{split}
Pr[win]=\sum^n_{i=1}\left(Pr[selects \ P_i]*\frac{|P_i|D_i(S=s)}{\sum^n_{k=1}|P_k|D_k(S=s)}\right) =\\
\frac{1}{\sum^n_{k=1}|P_k|D_k(S=s)}\sum^n_{i=1}\left(Pr[selects \ P_i]*|P_i|D_i(S=s)\right)
\end{split}
\end{equation}

An adversary with no prior knowledge of the distribution of instructions, selects programs uniformly with the probability $1/n$, Therefore:

\begin{equation}
P[win]=\frac{1}{\sum^n_{k=1}|P_k|D_k(S=s)}\sum^n_{i=1}\left(\frac{1}{n}*|P_i|D_i(S=s)\right)=\frac{1}{n}
\end{equation}

If the adversary has knowledge of the distributions, she can make educated guesses when selecting programs. In this case, the adversary can always choose the program $P_t$ that maximizes the value of $|P_i| D_i(S=s)$. In other words:
\begin{equation}
Pr[selects \ P_i]=
\begin{cases}
    \text{1} &  i=t \\
    \text{0} & o.w
\end{cases}
\end{equation}

With this strategy, the probability Equation \ref{eq:maingame} simplifies to the following formula:
 
 \begin{equation}
Pr[win]=\frac{|P_t|D_t(S=s)}{\sum^n_{k=1}|P_k|D_k(S=s)}
\end{equation}

In this case, the advantage of the adversary over random guessing can be expressed by the following non-negligible value:

\begin{equation}
\label{eq:genAdv}
\begin{aligned}
\epsilon &= \frac{|P_t|D_t(S=s)}{\sum^n_{k=1}|P_k|D_k(S=s)} - \frac{1}{n} \\
&= \frac{n|P_t|D_t(S=s)-\sum^n_{k=1}|P_k|D_k(S=s)}{n \sum^n_{k=1}|P_k|D_k(S=s)} \\ 
&\geq \frac{n|P_t|D_t(S=s)-\sum^n_{k=1}|P_k|D_k(S=s)}{n^2|P_t|D_t(S=s)} \\
&\geq \frac{n|P_t|D_t(S=s)-|P_t|D_t(S=s)}{n^2|P_t|D_t(S=s)} \\
&= \frac{n-1}{n^2}
\end{aligned}
\end{equation}

The non-negligible advantage of correctly guessing the provenance of a randomly selected instruction contradicts Theorem \ref{th:singleinst}, which assumes that the adversary has no advantage in determining the origin of a randomly chosen instruction.

To make the adversary's advantage $\epsilon$ negligible, input programs must have identical distributions and almost similar sizes. Identical distributions for input programs simplifies the adversary advantage to the following formula which is still a fraction of two polynomials and non-negiligible:
\begin{equation}
\label{eq:sameDistAdvant}
\begin{aligned}
\epsilon &= \frac{|P_t|D(S=s)}{\sum^n_{k=1}|P_k|D(S=s)} - \frac{1}{n} \\
&= \frac{|P_t|}{\sum^n_{k=1}|P_k|}-\frac{1}{n}=\frac{n|P_t|-\sum^n_{k=1}|P_k|}{n\sum^n_{k=1}|P_k|}\\
\end{aligned}
\end{equation}

In Equation \ref{eq:sameDistAdvant} if the length of all input programs are almost the same then $\epsilon$ converges to zero.

\end{document}